\pdfoutput=1
\RequirePackage{ifpdf}
\ifpdf
\documentclass[pdftex]{sigma}
\else
\documentclass{sigma}
\fi

\numberwithin{equation}{section}

\newtheorem{Theorem}{Theorem}[section]
\newtheorem{Lemma}[Theorem]{Lemma}
\newtheorem{Proposition}[Theorem]{Proposition}
 { \theoremstyle{definition}
\newtheorem{Remark}[Theorem]{Remark} }

\usepackage{tikz}
\usetikzlibrary{shapes,positioning,arrows,chains,matrix,scopes,fit,decorations.markings,decorations,automata}
\tikzstyle{block}=[draw opacity=0.7,line width=1.4cm]
\tikzset{
 state/.style={
 rectangle,
 rounded corners,
 draw=black, very thick,
 minimum height=2em,
 inner sep=2pt,
 text centered,
 },
}

\newcommand{\set}[1]{\left\{#1\right\}}
\newcommand{\cI}{\mathcal{I}}
\newcommand{\cL}{\mathcal{L}}
\newcommand{\cV}{\mathcal{V}}
\newcommand{\cG}{\mathcal{G}}

\newcommand{\cK}{\mathcal{K}}

\newcommand{\cA}{\mathcal{A}}

\newcommand{\tf}{\widetilde{f}}
\newcommand{\tg}{\widetilde{g}}

\newcommand{\PB}{\left\{\cdot ,\cdot\right\}}

\newcommand{\pb}[1]{\left\{#1\right\}}

\begin{document}
\allowdisplaybreaks

\newcommand{\arXivNumber}{1804.02564}

\renewcommand{\PaperNumber}{059}

\FirstPageHeading

\ShortArticleName{Dressing the Dressing Chain}

\ArticleName{Dressing the Dressing Chain}

\Author{Charalampos A.~EVRIPIDOU~$^\dag$, Peter H.~VAN DER KAMP~$^\dag$ and Cheng ZHANG~$^\ddag$}

\AuthorNameForHeading{C.A.~Evripidou, P.H.~van der Kamp and C.~Zhang}

\Address{$^\dag$~Department of Mathematics and Statistics, La Trobe University, \\
\hphantom{$^\dag$}~Melbourne, Victoria 3086, Australia}
\EmailD{\href{mailto:c.evripidou@latrobe.edu.au}{C.Evripidou@latrobe.edu.au}, \href{mailto:P.VanDerKamp@latrobe.edu.au}{P.VanDerKamp@latrobe.edu.au}}

\Address{$^\ddag$~Department of Mathematics, Shanghai University, 99 Shangda Road, Shanghai 200444, China}
\EmailD{\href{mailto:ch.zhang.maths@gmail.com}{ch.zhang.maths@gmail.com}}

\ArticleDates{Received April 18, 2018, in final form June 04, 2018; Published online June 15, 2018}

\Abstract{The dressing chain is derived by applying Darboux transformations to the spectral problem of the Korteweg--de Vries (KdV) equation. It is also an auto-B\"acklund transformation for the modified KdV equation. We show that by applying Darboux transformations to the spectral problem of the dressing chain one obtains the lattice KdV equation as the dressing chain of the dressing chain and, that the lattice KdV equation also arises as an auto-B\"acklund transformation for a modified dressing chain. In analogy to the results obtained for the dressing chain (Veselov and Shabat proved complete integrability for odd dimensional periodic reductions), we
study the $(0,n)$-periodic reduction of the lattice KdV equation, which is a~two-valued correspondence. We provide explicit formulas for its branches and establish complete integrability for odd~$n$.}


\Classification{35Q53; 37K05; 39A14}

\section{Introduction}
\looseness=1 The dressing chain \cite{Sha,ShaYa,ves_shab} appeared in the application of Darboux transformations to the Sch\"odinger (Sturm--Liouville) equation, which is the spectral problem for the Korteweg--de Vries (KdV) equation. A detailed study concerning the integrability properties as well as solutions of the model were presented in \cite{ves_shab}. In particular, the authors proved that the dressing chain with a~periodic constraint in odd dimensions is completely integrable in the sense of Liouville--Arnold.

The dressing chain can also be obtained as an auto-B\"acklund transformation of the modified KdV (mKdV) equation via the celebrated Miura transformation between KdV and mKdV~\cite{Miu}, and a symmetry of mKdV. The two ways of deriving the dressing chain are not unrelated, as the Miura transformation itself can be derived from factorisation of the Schr\"odinger equation~\cite{FG}.

Both the Darboux and B\"acklund approach can be seen as a discretisation process, and the two methods have been applied to other equations. In particular, in the discrete setting, Spiridonov and Zhedanov \cite{SZ} considered a tri-diagonal discrete Schr\"odinger equation, for which discrete Darboux transformations gave rise to two equivalent systems: the discrete time Toda lattice and a system they called the discrete dressing chain \cite[equations~(5.30) and~(5.31)]{SZ}. As the discrete Schr\"odinger equation considered in~\cite{SZ} is the spectral problem for the Toda lattice~\cite{MatSal} one could refer to these systems as dressing chains of the Toda lattice. Starting from the Volterra equation and using two discrete Miura transformations, Levi and Yamilov obtained an integrable lattice equation which they regard as a direct analogue of the dressing chain \cite[equation~(31)]{LY}, cf.~\cite{GHY}. In our context, we would refer to that equation as an auto-B\"acklund transformation for a modified Volterra equation.

In general, for a given integrable equation, one can ask the following questions, see Fig.~\ref{sch}:
\begin{enumerate}\itemsep=0pt
\item Is there a Miura or, more generally, a B\"acklund transformation to a modified equation which has a symmetry? This then gives rise to an auto-B\"acklund transformation of the modified equation, which discretises the equation.
\item Is there an associated spectral problem, whose factorisation yields a dressing chain?
\item Does the B\"acklund transformation~(1) arise in the factorisation of the spectral problem~(2)?
\item Does the auto-B\"acklund transformation for the modified equation coincide with the dressing chain of the equation?
\item Can the original equation be recovered by appropriate continuum limits?
\item Is the auto-B\"acklund transformation/dressing chain integrable?
\end{enumerate}

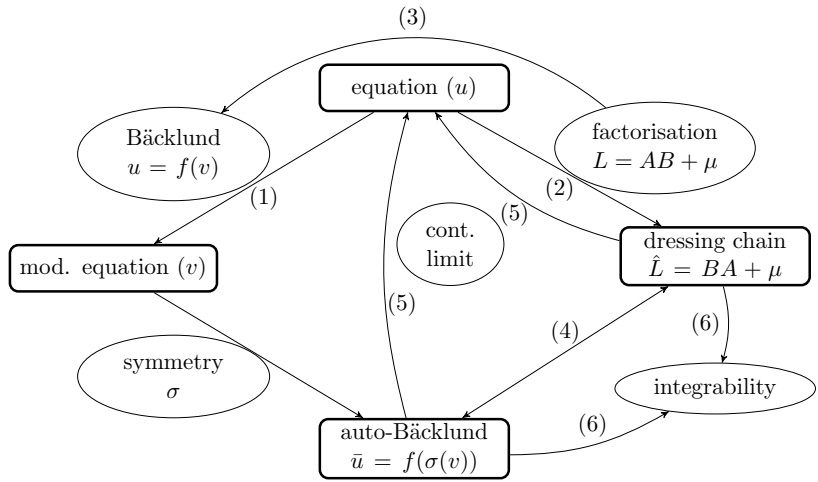
\begin{figure}[h]
\centering
\scalebox{0.8}{
\begin{tikzpicture}[->,>=stealth']

 \node[%
 state,
 text width=3.3cm] (modeq)
 {mod. equation $(v)$};

 \node[draw=black, ellipse, minimum width=3cm, align=center,
 right of = modeq, yshift=1.9 cm, text width = 2cm] (bactr)
 {B\"acklund $u=f(v)$};

 \node[draw=black, ellipse, minimum width=3cm, align=center,
 right of = modeq, yshift=-1.8 cm, text width = 2cm] (sym)
 {symmetry $\sigma$};

 % State: ACK with different content
 \node[state, 	% layout (defined above)
 text width=3cm, 	% max text width
 yshift=3cm, 		% move 2cm in y
 right of=modeq, 	% Position is to the right of QUERY
 node distance=5cm, 	% distance to QUERY
 anchor=center] (eq) 	% posistion relative to the center of the 'box'
 {%
 equation $(u)$
 };

 \node[draw=black, ellipse, minimum width=3cm, align=center,
 right of = eq, yshift=-1 cm, node distance=4cm, text width = 2.1cm] (fact)
 {factorisation $L=AB+\mu$};

 \node[draw=black, ellipse, minimum width=1cm, align=center,
 right of = eq, yshift=-5 cm, node distance = 5cm, text width = 2.1cm] (int)
 {integrability};

 \node[right of = eq, yshift=-3.6 cm, node distance = .55cm, text width = 2cm] (5)
 {(5)};

 \node[draw=black, ellipse, minimum width=1cm, align=center,
 right of = eq, yshift=-2.6 cm, node distance = .6cm, text width = 1cm] (cont)
 {cont. limit};

 % STATE QUERYREP
 \node[state,
 text width=3cm,
 yshift=-11cm,
 below of=eq,
 node distance=-5cm,
 anchor=center] (autob)
 {%
auto-B\"acklund
$\bar{u}=f(\sigma(v))$
 };

 % STATE EPC
 \node[state,
 text width=3cm,
 yshift=-2.8 cm,
 right of=eq,
 node distance=5cm,
 anchor=center
 ] (drch)
 {%
dressing chain\\
$\hat{L}=BA+\mu$
 };

 % draw the paths and print some Text below/above the graph
 \path

 (eq) 	edge node[anchor=south,below]{(1)} (modeq)
 (modeq) 	edge (autob)
 (eq) edge node[anchor=south,below]{(2)} (drch)
 (autob) edge[<->] node[anchor=south,above]{(4)} (drch)
 (fact) edge[bend right=40] node[anchor=south,above]{(3)} (bactr)
 (drch) edge[bend left=15,->] node[anchor=south,left]{(6)} (int)
 (autob) edge[bend right=15,->] node[anchor=south,above]{(6)} (int)
 (drch) edge[bend left=20] node[below]{(5)} (eq)
 (autob) edge[bend left=15] (eq)
 ;
% (autob) 	 edge node[anchor=left,right]{$SC_n = 0$} (eq);

\end{tikzpicture}}

\caption{Darboux and B\"acklund transformations.}\label{sch}
\end{figure}

In this paper, our starting point is the dressing chain. We factorise its discrete spectral problem, which itself is an exact discretisation, cf.~\cite{ZPZ}, of the (continuous) Schr\"odinger equation. It turns out that the {\em discrete} dressing chain (of the dressing chain) coincides with the (non-autonomous) lattice Korteweg--de Vries (lKdV) equation. By studying a related Lax representation we identify a~B\"acklund transformation to a modified dressing chain which admits a~symmetry. The derived auto-B\"acklund transformation is again given by the lKdV equation.

In analogy to the continuous case, cf.~\cite{ves_shab}, we study the (0,$n$)-periodic reduction\footnote{The ($n$,0)-periodic reduction gives rise to the same maps, up to a minus sign.} of the (discrete) dressing chain of the dressing chain (a.k.a.\ the lKdV equation), which is a two-valued correspondence (i.e., multi-valued map). We provide explicit formulas for its two branches, and establish linear growth of multi-valuedness. Moreover, we prove (in odd dimensions) that the map is Liouville integrable with respect to a quadratic Poisson structure of Lotka--Volterra type.

\section{Background} \label{sodc}
We clarify Fig.~\ref{sch} by succinctly providing some details for the KdV equation. We hope it also makes clear to the reader that how the dressing chain is related to the KdV equation is completely analogous to how the lattice KdV equation is related to the dressing chain.

The KdV equation $u_t=u_{xxx}-6uu_x$ arises as the compatibility condition, $L_t=[L,M]$, for the system of linear equations $L\phi =\lambda \phi$, $\phi_t=M\phi$ where $L$ is the Schr\"odinger operator $L=-D^2+u$, $M=4D^3-3(uD+Du)$, and $\lambda$ is a spectral parameter. One can check that if $v$ satisfies the mKdV equation $v_t=v_{xxx}-6(v^2+\alpha)v_x$, then $u$ given by the Miura transformation
\begin{gather}\label{Miu}
u=v_x+v^2+\alpha
\end{gather}
satisfies the KdV equation. As the mKdV equation is invariant under $v\rightarrow-v$ another Miura transformation is given by
\begin{gather}\label{Miub}
\bar{u}=-v_x+v^2+\alpha.
\end{gather}
Combining the two equations (\ref{Miu}) and (\ref{Miub}) yields an auto-B\"acklund transformation for the mKdV equation,
\begin{gather*}
(\bar{v}+v)_x=v^2-\bar{v}^2+\alpha-\bar{\alpha},
\end{gather*}
which coincides with the dressing chain \cite{ves_shab}. A related chain, which is an auto-B\"acklund transformation for the potential KdV equation, was already written down by Wahlquist and Estabrook~\cite{WE}, who used Bianchi's permutability theorem to show that it generates hierarchies of solutions due to a nonlinear superposition principle. More general auto-B\"acklund transformations (and their interpretation as differential-difference equations) were given in \cite{Levi2, Levi}. Auto-B\"acklund transformations for differential-difference equations are lattice equations, and some examples were presented in \cite{GY,LY}.

Darboux transformations for differential and difference equations (also known as dressing transformations) are maps of the functions and the coefficients that preserve the form of the equations \cite{Darb1}. They can be obtained by factorisation of operators, cf.~\cite{BC, FG, IH, Schr1, Schr2}, and they provide an effective way to construct exact solutions of a wide range of integrable equations (see, e.g., the monograph~\cite{Darb2}).

Recall that the Schr\"odinger operator $L$ can be decomposed as
\begin{gather*}
 L=-(D+v)(D-v)+\alpha ,
\end{gather*}
subject to the constraint (\ref{Miu}). Darboux \cite{Darb1} showed that under the transformation
\begin{gather}
\label{tip}
\phi\mapsto \widetilde{\phi}=(D-v) \phi ,
\end{gather}%
and $L \mapsto \widetilde{L}$, where (interchanging the two factors in the decomposition)
\begin{gather*}
 \widetilde{L} =-(D-v)(D+v)+\alpha
\end{gather*}
the form of the Sch\"odinger equation is unchanged: $\widetilde{L} \widetilde{\phi}=\lambda \widetilde{\phi}$. The $~\widetilde{}~$ operation characterises a~Darboux transformation for the Schr\"odinger equation $L$, if $\widetilde{L}$ is still a Schr\"odinger operator, i.e., $\widetilde{L} =-D^2+\widetilde{u}$ with $\widetilde{u} = -v_x+v^2+\alpha$, cf.\ equation~\eqref{Miub}. Iterated Darboux transformations result in the dependency of the functions $u$ and~$\phi$ on shifts in the $~\widetilde{}~$ direction, and~$\alpha$ becomes a~lattice parameter. Eliminating $u,\widetilde{u}$ in the above decompositions yields the dressing chain
\begin{gather}\label{dc}
(\widetilde{v}+v)_x=v^2-\widetilde{v}^2+\alpha-\widetilde{\alpha}.
\end{gather}%
Denoting $v=v_i$, $\widetilde{v}=v_{i+1}$, etc., and adding a periodic constraint, i.e., $v_{i+n}=v_{i}$ and $\alpha_{i+n}=\alpha_{i}$, one gets the finite dimensional systems of ordinary differential equations
\begin{gather}\label{pdc}
(v_{i+1}+v_i)_x=v_i^2-v_{i+1}^2+\alpha_i-\alpha_{i+1} , \qquad 1\leq i\leq n,
\end{gather}
which was shown to be completely integrable for odd $n$ \cite{ves_shab}.

\section{Dressing the dressing chain}
By eliminating the $x$-derivatives in the Schr\"odinger equation $L\phi=\lambda\phi$, using equation (\ref{tip}), one obtains the discrete Schr\"odinger equation
\begin{gather}\label{dSe}
K\phi=\lambda\phi,\qquad K=-T^2+h T+\alpha
\end{gather}
where
\begin{gather} \label{hv}
h=-v-\widetilde{v},
\end{gather}
and $T\colon z\rightarrow\widetilde{z}$ represents a shift operator. The discrete Schr\"odinger operator $K$ is the dual, with respect to~(\ref{tip}), to the continuous operator $L$, cf.~\cite{shabat1}. We note that the compatibility condition $K_x=[N,K]$, with $N=T+v$, provides a Lax representation for the dressing chain~(\ref{dc}).

The operator $K$ can be decomposed as, cf.~\cite{ZPZ},
\begin{gather*}
K = -(T+f) (T-g) - \beta .
\end{gather*}
Here $\beta$ does not depend on the $~\widetilde{}~$ direction (as $\alpha$ does) but will depend on another discrete direction introduced below. In order that such a~decomposition holds, one needs
\begin{gather} \label{h}
h=\widetilde{ g}-f
\end{gather}
and
\begin{gather} \label{fg}
f g=\alpha + \beta.
\end{gather}
Eliminating $f$ leads to $(h-\widetilde{ g}) g+\alpha+\beta=0$. This can be solved by posing $g = \widetilde{\psi}\psi^{-1}$, where~$\psi$ is a special solution of~\eqref{dSe} with $\lambda = -\beta$. Now that $f$ and $g$ are well defined, we can apply the usual tactics (interchanging the two factors in the decomposition) to generate a Darboux transformation for~\eqref{dSe}. With
\begin{gather}\label{hap}
\widehat{\phi}=G \phi,\qquad G=T-g ,
\end{gather}
and
\begin{gather*}
\widehat{K}=-(T-g)(T+f)-\beta
\end{gather*}
we have $\widehat{K} \widehat{\phi}=\lambda \widehat{\phi}$. Letting $\widehat{K} =- T^2+\widehat{h} T+\alpha$ imposes another constraint
\begin{gather} \label{hh}
\widehat{h}= g-\widetilde{f},
\end{gather}
which together with \eqref{h} and \eqref{fg} yields the non-autonomous lattice KdV equation (lKdV)
\begin{gather}\label{dkdv}
\tf-\widehat{f}=\frac{\alpha+\beta}{f}-\frac{\widetilde{\alpha}+\widehat{\beta}}{\widehat{\tf}} ,
\end{gather}
or, in terms of $g$,
\begin{gather}
\label{dkdvg}
g-\widehat{\widetilde{g}}=\frac{\widetilde{\alpha}+\beta}{\widetilde{g}}-\frac{\alpha+\widehat{\beta}}{\widehat{g}} .
\end{gather}
Shifts in the $~\widehat{}~$ direction correspond to a second discrete direction, created by iterated Darboux transformations, and the parameter $\beta$ varies in this direction. The linear system of equations~(\ref{dSe}) and~(\ref{hap}) provides a Lax representation for the lKdV equation, $\widehat{K}G=\widetilde{G}K$.

In the light of Fig.~\ref{sch}, equation (\ref{dkdv}), or (\ref{dkdvg}), is the dressing chain of the dressing chain. In analogy to the continuous case, we will consider a periodic reduction in the $~\widehat{}~$ direction, i.e., $f_{i+n}=f_i$ and $\beta_{i+n}=\beta_i$, and we take $\widetilde{\alpha}=\alpha$ to be a constant. The finite dimensional system of difference equations we will study is
\begin{gather}\label{dpdc}
\widetilde{f_i}-f_{i+1}=\frac{\alpha+\beta_i}{f_i}-\frac{\alpha+\beta_{i+1}}{\tf_{i+1}} ,\qquad 1\leq i\leq n.
\end{gather}%
As we make explicitly in Section~\ref{sc}, it gives rise to a two-valued correspondence.

It would also be justified to refer to the above system (\ref{dkdv}) as {\em the discrete dressing chain}, since its continuum limit coincides with \eqref{dc}. Using equations \eqref{h} and \eqref{hh}, one can express $f=\widehat{\widetilde{w}}-w$, $g=\widehat{w}-\widetilde{w}$. Substituting them into \eqref{fg} gives the lattice potential KdV equation $(\widehat{\widetilde{w}}-w)(\widehat{w}-\widetilde{w})=\alpha+\beta$, whose continuum limit with respect to the $~\widehat{ }~$ direction is \cite{HNJ1,Nij1}
\begin{gather}\label{pDC}
(\widetilde{w}+w)_x = (\widetilde{w}-w)^2 + \alpha .
\end{gather}
From \eqref{hv}, \eqref{h} and the above expressions for $f$ and $g$, one obtains $v=w-\widetilde{w}$ which relates the {\em potential dressing chain} \eqref{pDC} to the dressing chain \eqref{dc}. We will refer to the $(0,n)$-reduction of the lKdV equation (\ref{dpdc}) as the $n$-dimensional discrete dressing chain.

\section{The modified dressing chain}
One next wonders if the discrete dressing chain \eqref{dkdv} (or \eqref{dkdvg}) is the auto-B\"acklund transformation of a modified dressing chain. This is indeed the case. The Lax equation $G_x=\widehat{N}G-GN$ gives rise to the system $g(\widehat{v}-v)-g_x=\widetilde{v}-\widehat{v}-g+\widetilde{g}=0$, which together with $(h=)-v-\widetilde{v}=\widetilde{g}-f$ yields
\begin{gather}\label{fbt}
v=\frac12 \left(f-g-\frac{g_x}{g}\right),\qquad \widetilde{v}=\frac12 \left(f+g+\frac{g_x}{g}\right)-\widetilde{g}
\end{gather}
(as well as $\widehat{v}=\frac12 \big(f-g+\frac{g_x}{g}\big))$. The system (\ref{fbt}) provides a B\"acklund transformation, cf.~\cite[Definition~2.1.1]{HNJ1} between the dressing chain~\eqref{dc} and the following equation
\begin{gather}\label{gtx}
\tg+\frac{\widetilde{\alpha}+\beta}{\tg}-\frac{\tg_x}{\tg}=g+\frac{\alpha+\beta}{g}+\frac{g_x}{g} ,
\end{gather}
which we will refer to as {\em the modified dressing chain}. The modified dressing chain~(\ref{gtx}) admits the symmetry
\begin{gather*}
\sigma(g,\widetilde{g},x)=(\widetilde{g},g,-x).
\end{gather*}
Applying this symmetry to the right hand sides of (\ref{fbt}) and transforming the left hand sides by $(v,\widetilde{v}) \rightarrow (\widetilde{\bar{v}},\bar{v})$, we obtain another B\"acklund transformation
\begin{gather*} %\label{sbt}
\widetilde{\bar{v}}=\frac12 \left(\widetilde{f}-\widetilde{g}+\frac{\widetilde{g}_x}{\widetilde{g}}\right),\qquad
\bar{v}=\frac12 \left(\widetilde{f}+\widetilde{g}-\frac{\widetilde{g}_x}{\widetilde{g}}\right)-g.
\end{gather*}
Combining the two B\"acklund transformations $\big(\bar{v}+\widetilde{\bar{v}}=\overline{v+\widetilde{v}}\big)$ we obtain $\widetilde{f}-g=\bar{f}-\bar{\widetilde{g}}$, which shows that the lKdV equation is an auto-B\"acklund transformation for the modified dressing chain~(\ref{gtx}).

\section[Explicit formulas for the $n$-dimensional discrete dressing chain, and linear growth of multivaluedness]{Explicit formulas for the $\boldsymbol{n}$-dimensional discrete dressing\\ chain, and linear growth of multivaluedness}\label{sc}

In this section we consider the $(0,n)$-reduction of the lattice KdV equation, which is a~two-valued correspondence. We give explicit formulas for both branches ($M,N$), and prove that $MNM=N$. The latter implies that the $l$-th iteration of the correspondence is $2l$-valued, cf.\ \cite[Section~6.2]{KQ}.

In the finite reduction~\eqref{dpdc}, without loss of generality, we set $\alpha = 0$ since it can be absorbed into the parameters $\beta_j$. Having fixed $n\in\mathbb N$ and taking $i\in\cI=\set{1,2,\ldots,n}$ subject to the periodic boundary conditions $f_{n+i}=f_i$, $\beta_{n+i}=\beta_i$ for all $i\in\mathbb \cI$, the system of equations~\eqref{dpdc} reads
\begin{align}\label{ddc}
 \begin{split}
E_1\colon \quad & \tf_1 +\frac{\beta_2}{\tf_{2}} = f_{2} +\frac{\beta_1}{f_{1}} , \\
E_2\colon \quad & \tf_2 +\frac{\beta_3}{\tf_{3}} = f_{3} +\frac{\beta_2}{f_{2}} , \\
 & \qquad \vdots \quad \qquad \\
E_n\colon \quad & \tf_n +\frac{\beta_1}{\tf_{1}} = f_{1} +\frac{\beta_n}{f_{n}} .
\end{split}
\end{align}

These equations define a two-valued correspondence on $\mathbb R^n$. One solution of the system \eqref{ddc} is given by
\begin{gather} \label{os}
\tf_i=\frac{\beta_i}{f_i} ,\qquad i\in \cI
\end{gather}
(which is $\tf_i=g_i$, cf.~(\ref{fg})). This defines a map
\begin{gather*}%\label{trivial_map}
N\colon \ (f_1,f_2,\ldots,f_n)\mapsto \left(\frac{\beta_1}{f_1},\frac{\beta_2}{f_2},\ldots,\frac{\beta_n}{f_n}\right),
\end{gather*}
which is an involution. The other solution of the system \eqref{ddc} gives rise to a more intriguing map on $\mathbb R^n$, which will be denoted by $M$. We next provide explicit formulas for $M$ and for its inverse.

\begin{Remark} Consider the finite version of system \eqref{dpdc} defined by choosing $n\in\mathbb N$ and restricting $i\in\cI$ subject to the open boundary condition $f_{n+1}=\beta_{n+1}=0$. The resulting system then takes the form
\begin{gather*}
 \tf_i +\frac{\beta_{i+1}}{\tf_{i+1}} =f_{i+1} +\frac{\beta_i}{f_{i}} , \qquad \text{for } i=1,2,\ldots, n-1, \qquad \text{and} \qquad \tf_n =\frac{\beta_n}{f_{n}},
\end{gather*}
whose unique solution is given by the involution \eqref{os}.
\end{Remark}

In order to describe the nontrivial solution of \eqref{ddc} we introduce some notation. With \smash{$1<n\in\mathbb N$} and $k\in\cI$ we consider $\mathbb R^k$ with coordinates $f_1, f_2, \ldots, f_k$. We fix the parameters $\beta_1, \beta_2, \ldots, \beta_n$, and define functions $F_k\colon \mathbb R^{k}\rightarrow \mathbb R$ by $F_1(f_1)=1$ and
\begin{gather} \label{functionF}
F_k=F_k(f_1, f_2, \ldots, f_k)=f_1f_2^2f_3^2\cdots f_{k-1}^2f_k ,\qquad k>1.
\end{gather}
We also define a function $G:\mathbb R^{2n} \rightarrow \mathbb R$ by
\begin{align}\label{functionG}
\begin{split}
G&=G(f_1, f_2, \ldots, f_n,\beta_1,\beta_2,\ldots,\beta_n)\\
&=\sum_{i=0}^{n-1}
\left( \prod_{j=1}^i\beta_j \cdot F_{n-i}(f_{i+1}, f_{i+2}, \ldots, f_n)\right)\\
&=F_n(f_1, f_2, \ldots, f_n)+\beta_1F_{n-1}(f_2, f_3, \ldots, f_n)+\dots+\beta_1\beta_2\cdots\beta_{n-1}.
\end{split}
\end{align}
For $n=3$ the function $G$ reads
\begin{gather*}
G=f_1f_2^2f_3+\beta_1f_2f_3+\beta_1\beta_2.
\end{gather*}
We will make use of the following cyclic permutation $\tau\colon \cI\rightarrow\cI$
\begin{gather*}%\label{tau_def}
 \tau_i= \begin{cases}
 i+1, &\text{if } i< n, \\
 1 , &\text{if } i= n,
 \end{cases}
\end{gather*}
and of the involution $\sigma\colon \cI\rightarrow\cI$ defined by
\begin{gather*}%\label{sigma_def}
\sigma_i=n+1-i.
\end{gather*}
Simply stated, the permutation $\tau$ is a shift modulo $n$, and $\sigma$ is to reverse the elements of $\cI$. By some abuse of notation we will write, for any function $H$ depending on the variables $f_1, $ $f_2, \ldots, f_n$ and the parameters $\beta_1,\beta_2,\ldots,\beta_n$,
\begin{align*}
\tau H(f_1, f_2, \ldots, f_n,\beta_1,\beta_2,\ldots,\beta_n)&=H(f_{\tau_1}, f_{\tau_2}, \ldots, f_{\tau_n},\beta_{\tau_1},\beta_{\tau_2},\ldots,\beta_{\tau_n})\\
&=H(f_2, f_3, \ldots, f_1,\beta_2,\beta_3,\ldots,\beta_1) ,
\end{align*}
and similarly
\begin{align*}
\sigma H(f_1, f_2, \ldots, f_n,\beta_1,\beta_2,\ldots,\beta_n)&= H(f_{\sigma_1}, f_{\sigma_2}, \ldots, f_{\sigma_n},\beta_{\sigma_1},\beta_{\sigma_2},\ldots,\beta_{\sigma_n})\\
&=H(f_n, f_{n-1}, \ldots, f_1,\beta_n,\beta_{n-1},\ldots,\beta_1) .
\end{align*}
For example, for $n=3$ we have
\begin{gather*}
\tau G=f_1f_2f_3^2+\beta_2f_1f_3+\beta_2\beta_3 ,\qquad \sigma G=f_1f_2^2f_3+\beta_3f_1f_2+\beta_2\beta_3 .
\end{gather*}
A useful property of the above-defined functions is that, for any $n$, the expression
\begin{gather} \label{fixedpoint}
f_1f_nG-\beta_1\tau G=f_1^2f_2^2\cdots f_n^2-\beta_1\beta_2\cdots\beta_n
\end{gather}
is invariant under $\tau$ and $\sigma$. The following formula, which can be easily proved, is also useful. For any function $H$ depending on $f_1, f_2, \ldots, f_n$, $\beta_1,\beta_2,\ldots,\beta_n$, we have $\tau\sigma\tau H=\sigma H$, which implies, for all $i$,
\begin{gather}\label{prod_tau_sigma}
\sigma\tau^{i} H=\tau^{n-i}\sigma H .
\end{gather}

We now define functions
\begin{gather}\label{nontrivial_map}
M_i=f_{i-1}\frac{\tau^{i-1}G}{\tau^iG},
\end{gather}
where the indices are considered modulo $n$ and in the set $\cI$, in particular
\begin{gather*}
M_1=f_n\dfrac{G}{\tau G}.
\end{gather*}

\begin{Lemma}The map $M\colon (f_1,f_2,\ldots,f_n)\mapsto(M_1,M_2,\ldots,M_n)$, where $M_i$ is defined by \eqref{nontrivial_map}, is a solution of the system \eqref{ddc}.
\end{Lemma}

\begin{proof}By definition \eqref{nontrivial_map}, we have $M_i=\tau M_{i-1}$ for all $i$ and a similar property holds for the equations in \eqref{ddc} (each equation is obtained by applying $\tau$ to the previous one). Hence, it is enough to show that the functions $M_i$ satisfy the equation $E_1$. Taking $\tf_1=M_1$ and $\tf_2=M_2$, we have to verify
\begin{align*}
M_1 +\frac{\beta_2}{M_2} = f_{2} +\frac{\beta_1}{f_{1}} \iff& f_n\frac{G}{\tau G} +\frac{\beta_2}{f_1}\frac{\tau^2 G}{\tau G} = f_2+\frac{\beta_1}{f_1} \\
\iff& f_1f_nG+\beta_2\tau^2 G = (f_1f_2+\beta_1)\tau G \\
\iff& f_1f_nG-\beta_1\tau G= \tau(f_1f_n G-\beta_1\tau G)
\end{align*}
or equivalently that $f_1f_nG-\beta_1\tau G$ is fixed under $\tau$, which holds due to \eqref{fixedpoint}.
\end{proof}

The maps $M$ and $N$ satisfy the following relation.
\begin{Lemma} \label{lem}The map $N$ is a reversing symmetry of $M$.
\end{Lemma}
\begin{proof}The statement entails
\begin{gather*}
MN=NM^{-1}.
\end{gather*}
Applying the involution $\sigma$ to all indices in system \eqref{ddc}, and interchanging $f_i\leftrightarrow \tf_i$ one observes that
\begin{gather*}
E_j\big(f_i,\tf_i\big)=E_{n-j}\big(\tf_{\sigma_i},f_{\sigma_i}\big),
\end{gather*}
for $1\leq j< n$ and $E_n\big(f_i,\tf_i\big)=E_n\big(\tf_{\sigma_i},f_{\sigma_i}\big)$. If we write the nontrivial solution as $\tf_j=M_j(f_i)$, then we also have $f_{\sigma_j}=M_j \big(\tf_{\sigma_i}\big)$, and applying $\sigma$ to the $j$ index (which just enumerates the functions), one has
\begin{gather*}
f_j=M_{\sigma_j}\big(\tf_{\sigma_i}\big)=M_j^{-1}\big(\tf_i\big).
\end{gather*}
Thus the inverse map is
\begin{gather*} %\label{eim}
M^{-1}=\sigma M \sigma,
\end{gather*}
which as a function of the $f_i$ has components
\begin{gather*}
M^{-1}_j = \sigma \left( f_{\sigma_j-1} \frac{ \tau^{\sigma_j-1}G}{\tau^{\sigma_j}G}\right) = f_{j+1} \frac{\sigma \tau^{n-j}G}{\sigma\tau^{n-j+1}G} = f_{j+1} \frac{\tau^{j}\sigma G}{\tau^{j-1}\sigma G}.
\end{gather*}
The latter formula is obtained using~\eqref{prod_tau_sigma}. It follows immediately that $M^{-1}_{j+1}=\tau M^{-1}_j$. Therefore it is enough to show that $(MN)_1=\big(NM^{-1}\big)_1$, that is
\begin{gather}\label{to_be_proved}
\frac{\beta_n G\big(\frac{\beta_1}{f_1},\ldots,\frac{\beta_n}{f_n}\big)}{f_n \tau G\big(\frac{\beta_1}{f_1},\ldots,\frac{\beta_n}{f_n}\big)}
= \frac{\beta_1\sigma G(f_1,\ldots,f_n)}{f_2 \tau\sigma G(f_1,\ldots,f_n)}.
\end{gather}
Using the formulas \eqref{functionF} and \eqref{functionG}, one has
\begin{align*}
G\left(\frac{\beta_1}{f_1},\ldots,\frac{\beta_n}{f_n}\right)&=
\sum_{i=0}^{n-1} \left(\prod_{j=1}^i\beta_j\!\right)F_{n-i}\left(\frac{\beta_{i+1}}{f_{i+1}},\ldots,\frac{\beta_n}{f_n}\right) =\sum_{i=0}^{n-1}%
\left(\prod_{j=1}^i\beta_j\!\right)\frac{F_{n-i}(\beta_{i+1},\ldots,\beta_n)}{F_{n-i}(f_{i+1},\ldots,f_n)}\\
&=\frac{\prod\limits_{j=1}^{n-1}\beta_j}{F_n(f_1,\ldots,f_n)}\sum_{i=0}^{n-1}\left(\prod_{j=i+2}^{n}\beta_j\right)F_{i+1}(f_1,\ldots,f_{i+1}) ,
\end{align*}
and similarly
\begin{gather*}
\tau G\left(\frac{\beta_1}{f_1},\ldots,\frac{\beta_n}{f_n}\right) = \frac{\prod\limits_{j=1}^{n-1}\beta_{j+1}}{F_n(f_2,\ldots,f_n,f_1)}
\sum_{i=0}^{n-1}\left(\prod_{j=i+2}^{n}\beta_{j+1}\right)F_{i+1}(f_2,\ldots,f_{i+2}),\\
\tau\sigma G(f_1,\ldots,f_n) =\sum_{i=0}^{n-1} \left(\prod_{j=1}^i\beta_{n+2-j}\right)F_{n-i}(f_{i+2},f_{i+1},\ldots,f_2) ,
\end{gather*}
and
\begin{gather*}
\sigma G(f_1,\ldots,f_n) =\tau\sigma\tau G(f_1,\ldots,f_n) =\sum_{i=0}^{n-1}\left(\prod_{j=1}^i\beta_{n+1-j}\right)F_{n-i}(f_{i+1},f_{i},\ldots,f_1) .
\end{gather*}
Combining all these leads to \eqref{to_be_proved}.
\end{proof}

As a corollary of Lemma \ref{lem}, using the fact that $N$ is an involution, it follows that $MN$ and $NM$ are involutions, and hence $M$ can be written as a composition of two involutions, $M=(MN)N=N(NM)$. Furthermore, Lemma~\ref{lem} implies the relations: $MNM=NNN$ and $MNN=NNM$. These relations are also satisfied by the branches of the quotient-difference $(n, 0)$-correspondence. In \cite[Section~6.2]{KQ} it is proved that the $l$-th iteration of such a correspondence is $2l$-valued.

\section[Complete integrability of the odd-dimensional discrete dressing chain]{Complete integrability of the odd-dimensional discrete\\ dressing chain} \label{fs}
In this section we show that the correspondence $(M,N)$ is Liouville integrable with respect to a quadratic Poisson structure which is of Lotka--Volterra type. Our main result is the following theorem.
\begin{Theorem}\label{thr:ddc_int}For odd $n$, the correspondence defined by \eqref{ddc} is Liouville integrable.
\end{Theorem}

Before proving the theorem, we introduce the relevant Poisson structures and provide some of their basic properties.

\subsection{Lotka--Volterra Poisson structures}
Lotka--Volterra Poisson structures are homogeneous quadratic, and they are defined on $\mathbb R^n$ by the formulas%
\begin{gather}\label{LVPB}
\pb{x_i,x_j}_q=A_{i,j}x_ix_j,\qquad i,j\in\cI ,
\end{gather}
where $A$ is a constant skew-symmetric matrix.

The rank of this Poisson structure is equal, at a generic point, to the rank of the constant matrix $A$. Each null-vector of the matrix~$A$ is associated to a Casimir of the corresponding Poisson bracket. If $\textbf{v}=(v_1,v_2,\ldots,v_n)$ is such that $\textbf{v}A=0$, then the function $\prod\limits_{i=1}^nx_i^{v_i}$ is a Casimir of the Poisson bracket. Two linearly independent null-vectors correspond to two functionally independent Casimirs (for a proof see \cite[Example~8.14]{polpoisson}).

In what follows we consider the Lotka--Volterra structures \eqref{LVPB} where $A$ is the $n\times n$ skew-symmetric matrix with its upper triangular part defined by
\begin{gather}\label{eq:Bogo_mat}
 A_{i,j}=(-1)^{j-i+1} ,\qquad 1\leq i<j \leq n .
\end{gather}
The rank of the matrix $A$ is $n$ when $n$ is even and $n-1$ when $n$ is odd with the null-vector $\textbf{v}=(1,1,\ldots,1)$; a Casimir of the corresponding Poisson structure is the function $x_1x_2\cdots x_n$. With $H=x_1+x_2+\cdots+x_n$, the Hamiltonian vector field $\{\cdot, H\}_q$ defines a system of differential equations which, up to a simple change of variables, is isomorphic to the Bogoyavlenskij lattice \cite{bog2, bog3, damianou_evr_kass_van, pol2014}.

A simpler Poisson structure is the constant Poisson structure defined by the brackets
\begin{gather}\label{CPB}
\pb{x_i,x_j}_c=B_{i,j} ,
\end{gather}
where $B$ is a constant $n\times n$ skew-symmetric matrix. After some tedious but straightforward calculations (see \cite[Proposition 3]{PPP2017}), one proves that the Poisson structures \eqref{LVPB} and \eqref{CPB} are compatible if and only if
\begin{gather}\label{eq:expl_mat_B}
\begin{cases}
B_{i,j}=0 & \text{for all } |j-i|>1, \text{ when } n \text{ is even,}\\
B_{i,j}=0 & \text{for all } 1<|j-i|<n-1, \text{ when } n \text{ is odd.}
\end{cases}
\end{gather}

For $n$ odd, let $\PB_{\bf{b}}=\PB_q+\PB_c$ denote the sum of the brackets defined by (\ref{LVPB}), (\ref{eq:Bogo_mat}) and (\ref{CPB}), (\ref{eq:expl_mat_B}) where the subscript $\bf{b}$ is the vector ${\bf b}=(b_{1,2},b_{2,3},\ldots,b_{n-1,n}, b_{1,n})$, i.e., the non-zero elements of the matrix $B$. With $H=x_1+x_2+\cdots+x_n$, the Hamiltonian vector field $\{\cdot, H\}_{\bf{b}}$ defines a system of differential equations which can be transformed to the dressing chain~\eqref{pdc}. The Poisson structure $\PB_{\bf b}$ is of rank $n-1$ but with a more complicated Casimir than the product $x_1x_2\cdots x_n$ (see \cite{PPP2017, ForHon} for an explicit construction of this Casimir).

\subsection{Complete integrability}

We start by showing that the maps $M$ and $N$, defined in Section \ref{sc}, preserve the Poisson structure $\PB_{\bf b}=\PB_q+\PB_c$ with ${\bf b}=(-\beta_2,-\beta_3,\ldots,-\beta_n,\beta_1)$. To this end we define two additional maps $\phi,\psi\colon \mathbb R^n\rightarrow\mathbb R^n$,
\begin{align*}
\phi(f_1,f_2,\ldots,f_n)&=(f_1+g_2,f_2+g_3,\ldots,f_n+g_1) ,\\
\psi(f_1,f_2,\ldots,f_n)&=(f_2+g_1,f_3+g_2,\ldots,f_1+g_n) .
\end{align*}
\begin{Lemma}
\label{poissonmaps}
For $n$ odd, the maps $M$, $N$, $\phi$ and $\psi$ are Poisson maps as follows
\begin{enumerate}\itemsep=0pt
\item[$1)$] $\phi\colon (\mathbb R^n,\PB_q)\rightarrow (\mathbb R^n,\PB_{{\bf b}})$;
\item[$2)$] $\psi\colon (\mathbb R^n,\PB_q)\rightarrow (\mathbb R^n,\PB_{{\bf b}})$;
\item[$3)$] $M\colon (\mathbb R^n,\PB_q)\rightarrow (\mathbb R^n,\PB_q)$;
\item[$4)$] $N\colon (\mathbb R^n,\PB_q)\rightarrow (\mathbb R^n,\PB_q)$.
\end{enumerate}
\end{Lemma}
\begin{proof}The proof of items $(1)$, $(2)$ and $(4)$ follows from straightforward computations. For example, for any $i=1,2,\ldots, n$
\begin{align*}
\pb{f_i,f_{i+1}}_{\bf b}\circ \phi& =(f_i+g_{i+1})(f_{i+1}+g_{i+2})-
\beta_{i+1}=f_if_{i+1}+f_ig_{i+2}+g_{i+1}g_{i+2}\\
& =\{f_i+g_{i+1},f_{i+1}+g_{i+2}\}_q ,\\
\pb{f_i,f_{i+1}}_{\bf b}\circ \psi& =(f_{i+1}+g_i)(f_{i+2}+g_{i+1})-\beta_{i+1}=f_{i+1}f_{i+2}+g_if_{i+2}+g_ig_{i+1}\\
&=\pb{f_{i+1}+g_i,f_{i+2}+g_{i+1}}_q ,
\end{align*}
where the indices are considered modulo $n$ and in the set $\cI$. Item (3) follows from items~(1) and~(2) combined with
\begin{gather*}
\phi\circ M=\psi \qquad \text{and} \qquad \phi\circ N=\psi,
\end{gather*}
which are a consequence of \eqref{ddc}.
\end{proof}

\begin{Remark} The involution $N$ preserves any Poisson structure of Lotka--Volterra form. It follows that item (4), of the previous proposition, is (trivially) true for all $n\in\mathbb N$. However, this is not the case for items $(1)-(3)$ as for even $n$ the equations in the previous proof do not hold for $i=n$. Note, for even $n$ the bracket $\PB_{{\bf b}}$ is not Poisson, see condition~\eqref{eq:expl_mat_B}.
\end{Remark}

To derive sufficiently many independent invariants for the map \eqref{nontrivial_map}, we employ a matrix version of the Lax representation. Let $\widetilde{\Phi_i}=\cV_i \Phi_i$, $\widehat{\Phi_i}=\Phi_{i+1}=\cG_i \Phi_i$, where
\begin{gather*} %\label{laxm}
\cV_i=
\begin{pmatrix}
 0 & 1 \\
 -\lambda & \widetilde{g_i}-f_i
\end{pmatrix},
\qquad
\cG_i=
\begin{pmatrix}
 -g_i & 1 \\
 -\lambda & -f_i
\end{pmatrix}
\end{gather*}
and $\Phi=\big(\phi,\widetilde{\phi}\big)^T$. The compatibility condition is now written (modulo the periodic condition $n+j=j$) as $\cV_{i+1} \cG_i=\widetilde{\cG}_i \cV_i$. Then one has
\begin{gather*}
\widetilde{\cG}_n \widetilde{\cG}_{n-1} \cdots \widetilde{\cG}_1=\cV_1 \cG_n \cG_{n-1} \cdots \cG_1 \cV_1^{-1} ,
\end{gather*}
which implies that the (monodromy) matrix $\cL=\cG_n \cG_{n-1} \cdots \cG_1$ is isospectral. The eigenvalues of $\cL $ (and therefore its trace) are invariants of our map~\eqref{nontrivial_map}.

Using a decomposition of the Lax matrix $\cG_i$ similar to the one used in \cite{PPP2017, ves_shab} (explained in a more general manner in \cite{tran}), we are able to calculate the trace of $\cL$ and provide a succinct formula for the invariants.  Futhermore, using the results of \cite{ves_shab} we show that these invariants are in involution with respect to the Poisson bracket \eqref{LVPB} with matrix \eqref{eq:Bogo_mat} and therefore complete the proof of Theorem~\ref{thr:ddc_int}.

\begin{Proposition}Let $r=\frac{n-1}{2}\in\mathbb{N}$. There exist functions $I_0, I_1, \ldots, I_r$, which are polynomials of the variables
\begin{gather*}
x_1=-(g_2+f_1),\quad x_2=-(g_3+f_2),\quad \ldots,\quad x_n=-(g_1+f_n)
\end{gather*}
of degree $\deg(I_i)=2i+1$, such that $\operatorname{tr}(\cL)=\sum\limits_{i=0}^r I_i\lambda^i .$
With
\begin{gather*}
D_i=\frac{\partial^2}{\partial x_i \partial x_{i+1}},\qquad D=\sum_{i\in\cI} D_i ,
\end{gather*}
where $x_{n+1}=x_1$ is understood, we have
\begin{gather*}
I_0=\prod_{i\in\cI} (1-\beta_i D_i) \prod_{j\in\cI} x_j , \qquad I_j=\frac{(-1)^j}{j!}D^j I_{0} .
\end{gather*}
Furthermore, these functions are in involution with respect to the Poisson bracket $\PB_q$.
\end{Proposition}

\begin{proof}We decompose the matrix $\cG_i$ as follows
\begin{gather*}
\cG_i=\cA^i_i-(\lambda+\beta_i)\cK ,
\end{gather*}
where
\begin{gather*}
\cK=
\begin{pmatrix}
0 & 0 \\
1 & 0
\end{pmatrix}, \qquad
\cA^i_j=\begin{pmatrix}
-1 \\
f_i
\end{pmatrix}
\cdot \begin{pmatrix}
g_j & -1
\end{pmatrix}.
\end{gather*}%
This structure simplifies the computations. For example we can immediately verify the following properties
\begin{gather*}
\cK^2=\textbf{0},\qquad \cA^i_j \cA^k_l=-(g_j+f_k) \cA^i_l,\qquad \cA^i_j \cK \cA^k_l=\cA^i_l,\qquad \cK \cA^i_j \cK=\cK.
\end{gather*}
The monodromy matrix is, with $y_i=-(\lambda+\beta_i)$, in the form
\begin{align*}
\cL&=\big(\cA^n_n+y_n\cK\big)\big(\cA^{n-1}_{n-1}+y_{n-1}\cK\big)\cdots\big(\cA^1_1+y_1\cK\big)\\
&=\cA^n_n\cA^{n-1}_{n-1}\cdots \cA^1_1 + \sum_{i\in\cI} y_i\cA^n_n\cA^{n-1}_{n-1}\cdots \cA^{i+1}_{i+1}\cK \cA^{i-1}_{i-1}\cdots \cA^1_1 + \sum_{i,j\in\cI} \cdots\\
&=x_nx_{n-1}\cdots x_2\cA^n_1 + \sum_{i\in\cI} y_i x_nx_{n-1}\cdots x_{i+2}x_{i-1}\cdots x_1\cA^n_1 + \sum_{i,j\in\cI} \cdots,
\end{align*}
and its trace
\begin{align*}
\operatorname{tr}(\cL)&=\prod_{i\in\cI} x_i + \sum_{i\in\cI}\frac{y_i}{x_{i+1}x_{i}}\prod_{j\in\cI} x_j + \sum_{i>j+1\in\cI}\frac{y_i}{x_{i+1}x_{i}} \frac{y_j}{x_{j+1}x_{j}} \prod_{k\in\cI} x_j + \cdots \\
&=\prod_{i\in\cI} \left(1+y_iD_i\right) \prod_{j\in\cI}x_i=\sum_{i=0}^k I_i\lambda^i.
\end{align*}
Substituting $\lambda=0$ yields the expression for $I_0$, i.e.,
\begin{gather}\label{I0}
I_0=(1-\beta_1D_1)(1-\beta_2D_2)\cdots(1-\beta_nD_n)\prod_{i\in\cI} x_i ,
\end{gather}
while
\begin{align*}
I_1&=\sum_{j\in\cI}(1-\beta_1D_1)(1-\beta_2D_2)\cdots(-D_j)\cdots(1-\beta_nD_n)\prod_{i\in\cI} x_i ,\\
I_2&=\sum_{j<k\in\cI}(1-\beta_1D_1)(1-\beta_2D_2)\cdots(-D_j)\cdots(-D_k)\cdots(1-\beta_nD_n)\prod_{i\in\cI} x_i ,\\
&\vdots
\end{align*}
Applying $D$ to $I_0$ gives a sum of similar products where the $j$-th term $(1-\beta_jD_j)$ is replaced by $D_j$, which is equal to $-I_1$. Similarly, applying this operator $k$ times yields $(-1)^k I_k$ $k!$ times, namely once for each permutation of $k$ indices. This proves the first part of the proposition.

To prove the second part of the proposition, we first note that the invariants $I_i$ (as functions of $x_i$), as obtained from the trace of $\cL$, coincide with the invariants that Veselov and Shabat provided for the continuous dressing chain~\cite{ves_shab}. In that paper, using the Lenard--Magri scheme, they proved that the invariants $I_i$ are in involution with respect to the Poisson bracket $\PB_{\bf b}=\PB_q+\PB_c$. Therefore, in order to prove that the invariants $I_i$ (considered as functions of $f_i$) are in involution with respect to $\PB_q$, it suffices to show that the map
\begin{gather*}
\mathbb R^n \rightarrow \mathbb R^n, \qquad (f_1, f_2, \ldots, f_n)\mapsto (x_1, x_2, \ldots, x_n)
\end{gather*}
is a Poisson map between $(\mathbb R^n,\PB_q)$ and $(\mathbb R^n,\PB_{\bf b})$. This is precisely item $(1)$ of Lemma~\ref{poissonmaps}.
\end{proof}

\begin{Remark}The expression for $I_0$ (\ref{I0}) in terms of $f_i$, $g_i$, and $\beta_i$ is quite cumbersome, e.g., for $n=3$,
\begin{gather*}
I_0 =-f_{{1}}f_{{2}}f_{{3}}-f_{{1}}f_{{2}}g_{{1}}-f_{{1}}f_{{3}}g_{{3}}-f_{{1}}g_{{1}}g_{{3}}-f_{{2}}f_{{3}}g_{{2}}-f_{{2}}g_{{1}}g_{{2}}-f_{{3}}g_{{2}}g_{{3}}\\
\hphantom{I_0 =}{} -g_{{1}}g_{{2}}g_{{3}}+\beta_{{1}}f_{{2}}+\beta_{{3}}f_{{1}}+\beta_{{1}}g_{
{3}}+\beta_{{2}}f_{{3}}+\beta_{{2}}g_{{1}}+\beta_{{3}}g_{{2}} .
\end{gather*}
However, when imposing the relation $f_ig_i=\beta_i$, the expression simplifies drastically, and we have $I_0=-(f_1f_2f_3+g_1g_2g_3)$. The fact that a similar expression can be obtained for any $n$ can be seen from
\begin{gather*}
\cL\mid_{\lambda=0}=\prod_{i=1}^n\begin{pmatrix}-g_i&1\\0&-f_i\end{pmatrix}=(-1)^n
\begin{pmatrix}\displaystyle \prod_{i=1}^n g_i&\displaystyle -\sum_{i=1}^n\prod_{j=1}^{i-1}f_i\prod_{k=i+1}^{n}g_i\\0 &\displaystyle\prod_{i=1}^nf_i\end{pmatrix}.
\end{gather*}
\end{Remark}

\begin{Remark} For $n$ even the map $M$ is anti-volume preserving and for odd~$n$ it is volume preserving. The map~$N$ is measure preserving when~$n$ is even and anti-measure preserving when~$n$ is odd. The density of the measure is $\prod\limits_{i=1}^n\frac{1}{f_i}$.
\end{Remark}

\subsection*{Acknowledgements}
This work was supported by the Australian Research Council, by the China Strategy Implementation Grant Program of La Trobe University, by the NSFC (No.~11601312) and by the Shanghai Young Eastern Scholar program (2016-2019).

\pdfbookmark[1]{References}{ref}
\LastPageEnding


\begin{thebibliography}{99}
\footnotesize\itemsep=0pt

\bibitem{bog2}
Bogoyavlenskij O.I., Integrable discretizations of the {K}d{V} equation,
 \href{https://doi.org/10.1016/0375-9601(88)90542-7}{\textit{Phys. Lett.~A}} \textbf{134} (1988), 34--38.

\bibitem{bog3}
Bogoyavlenskij O.I., Integrable {L}otka--{V}olterra systems, \href{https://doi.org/10.1134/S1560354708060051}{\textit{Regul.
 Chaotic Dyn.}} \textbf{13} (2008), 543--556.

\bibitem{BC}
Burchnall J.L., Chaundy T.W., Commutative ordinary differential operators,
 \href{https://doi.org/10.1112/plms/s2-21.1.420}{\textit{Proc. London Math. Soc.}} \textbf{21} (1923), 420--440.

\bibitem{damianou_evr_kass_van}
Damianou P.A., Evripidou C.A., Kassotakis P., Vanhaecke P., Integrable
 reductions of the {B}ogoyavlenskij--{I}toh {L}otka--{V}olterra systems,
 \href{https://doi.org/10.1063/1.4978854}{\textit{J.~Math. Phys.}} \textbf{58} (2017), 032704, 17~pages,
 \href{https://arxiv.org/abs/1609.09507}{arXiv:1609.09507}.

\bibitem{Darb1}
Darboux G., Le\c{c}ons sur la th\'eorie g\'en\'erale des surfaces.~{I},~{II},
 Jacques Gabay, Sceaux, 1993.

\bibitem{PPP2017}
Evripidou C.A., Kassotakis P., Vanhaecke P., Integrable deformations of the
 {B}ogoyavlenskij--{I}toh {L}otka--{V}olterra systems, \href{https://doi.org/10.1134/S1560354717060090}{\textit{Regul. Chaotic
 Dyn.}} \textbf{22} (2017), 721--739, \href{https://arxiv.org/abs/1709.06763}{arXiv:1709.06763}.

\bibitem{FG}
Fordy A.P., Gibbons J., Factorization of operators. {I}.~{M}iura
 transformations, \href{https://doi.org/10.1063/1.524357}{\textit{J.~Math. Phys.}} \textbf{21} (1980), 2508--2510.

\bibitem{ForHon}
Fordy A.P., Hone A., Discrete integrable systems and {P}oisson algebras from
 cluster maps, \href{https://doi.org/10.1007/s00220-013-1867-y}{\textit{Comm. Math. Phys.}} \textbf{325} (2014), 527--584,
 \href{https://arxiv.org/abs/1207.6072}{arXiv:1207.6072}.

\bibitem{GHY}
Garifullin R.N., Habibullin I.T., Yamilov R.I., Peculiar symmetry structure of
 some known discrete nonautonomous equations, \href{https://doi.org/10.1088/1751-8113/48/23/235201}{\textit{J.~Phys.~A: Math.
 Theor.}} \textbf{48} (2015), 235201, 27~pages, \href{https://arxiv.org/abs/1501.05435}{arXiv:1501.05435}.

\bibitem{GY}
Garifullin R.N., Yamilov R.I., Integrable discrete nonautonomous quad-equations
 as {B}\"{a}cklund auto-transformations for known {V}olterra and {T}oda type
 semidiscrete equations, \href{https://doi.org/10.1088/1742-6596/621/1/012005}{\textit{J.~Phys. Conf. Ser.}} \textbf{621} (2015),
 012005, 18~pages, \href{https://arxiv.org/abs/1405.1835}{arXiv:1405.1835}.

\bibitem{HNJ1}
Hietarinta J., Joshi N., Nijhoff F.W., Discrete systems and integrability,
 \href{https://doi.org/10.1017/CBO9781107337411}{\textit{Cambridge Texts in Applied Mathematics}}, Cambridge University Press,
 Cambridge, 2016.

\bibitem{IH}
Infeld L., Hull T.E., The factorization method, \href{https://link.aps.org/doi/10.1103/RevModPhys.23.21}{\textit{Rev. Modern Phys.}}
 \textbf{23} (1951), 21--68.

\bibitem{polpoisson}
Laurent-Gengoux C., Pichereau A., Vanhaecke P., Poisson structures,
 \href{https://doi.org/10.1007/978-3-642-31090-4}{\textit{Grundlehren der Mathematischen Wissenschaften}}, Vol.~347, Springer,
 Heidelberg, 2013.

\bibitem{Levi2}
Levi D., Nonlinear differential-difference equations as {B}\"acklund
 transformations, \href{https://doi.org/10.1088/0305-4470/14/5/028}{\textit{J.~Phys.~A: Math. Gen.}} \textbf{14} (1981),
 1083--1098.

\bibitem{Levi}
Levi D., Benguria R., B\"acklund transformations and nonlinear differential
 difference equations, \href{https://doi.org/10.1073/pnas.77.9.5025}{\textit{Proc. Nat. Acad. Sci. USA}} \textbf{77} (1980),
 5025--5027.

\bibitem{LY}
Levi D., Yamilov R.I., The generalized symmetry method for discrete equations,
 \href{https://doi.org/10.1088/1751-8113/42/45/454012}{\textit{J.~Phys.~A: Math. Theor.}} \textbf{42} (2009), 454012, 18~pages,
 \href{https://arxiv.org/abs/0902.4421}{arXiv:0902.4421}.

\bibitem{MatSal}
Matveev V.B., Salle M.A., Differential-difference evolution equations.
 {II}.~{D}arboux transformation for the {T}oda lattice, \href{https://doi.org/10.1007/BF00397217}{\textit{Lett. Math.
 Phys.}} \textbf{3} (1979), 425--429.

\bibitem{Darb2}
Matveev V.B., Salle M.A., Darboux transformations and solitons, \href{https://doi.org/10.1007/978-3-662-00922-2}{\textit{Springer Series
 in Nonlinear Dynamics}}, Springer-Verlag, Berlin, 1991.

\bibitem{Miu}
Miura R.M., Korteweg--de {V}ries equation and generalizations.
 {I}.~{A}~remarkable explicit nonlinear transformation, \href{https://doi.org/10.1063/1.1664700}{\textit{J.~Math.
 Phys.}} \textbf{9} (1968), 1202--1204.

\bibitem{Nij1}
Nijhoff F., Capel H., The discrete {K}orteweg--de {V}ries equation,
 \href{https://doi.org/10.1007/BF00994631}{\textit{Acta Appl. Math.}} \textbf{39} (1995), 133--158.

\bibitem{Schr1}
Schr\"odinger E., A method of determining quantum-mechanical eigenvalues and
 eigenfunctions, \textit{Proc. Roy. Irish Acad. Sect.~A.} \textbf{46} (1940),
 9--16.

\bibitem{Schr2}
Schr\"odinger E., Further studies on solving eigenvalue problems by
 factorization, \textit{Proc. Roy. Irish Acad. Sect.~A.} \textbf{46} (1941),
 183--206.

\bibitem{Sha}
Shabat A., The infinite-dimensional dressing dynamical system, \href{https://doi.org/10.1088/0266-5611/8/2/009}{\textit{Inverse
 Problems}} \textbf{8} (1992), 303--308.

\bibitem{shabat1}
Shabat A., Dressing chains and lattices, in Proceedings of the {W}orkshop on
 {N}onlinearity, {I}ntegrability and {A}ll {T}hat: {T}wenty {Y}ears after
 {NEEDS} '79 ({G}allipoli, 1999), World Sci. Publ., River Edge, NJ, 2000,
 331--342.

\bibitem{ShaYa}
Shabat A.B., Yamilov R.I., Symmetries of nonlinear lattices, \textit{Leningrad
 Math.~J.} \textbf{2} (1991), 377--400.

\bibitem{SZ}
Spiridonov V., Zhedanov A., Discrete {D}arboux transformations, the
 discrete-time {T}oda lattice, and the {A}skey--{W}ilson polynomials,
 \href{https://doi.org/10.4310/MAA.1995.v2.n4.a1}{\textit{Methods Appl. Anal.}} \textbf{2} (1995), 369--398.

\bibitem{tran}
Tran D.T., van~der Kamp P.H., Quispel G.R.W., Closed-form expressions for
 integrals of traveling wave reductions of integrable lattice equations,
 \href{https://doi.org/10.1088/1751-8113/42/22/225201}{\textit{J.~Phys.~A: Math. Theor.}} \textbf{42} (2009), 225201, 20~pages.

\bibitem{pol2014}
van~der Kamp P.H., Kouloukas T.E., Quispel G.R.W., Tran D.T., Vanhaecke P.,
 Integrable and superintegrable systems associated with multi-sums of
 products, \href{https://doi.org/10.1098/rspa.2014.0481}{\textit{Proc.~R. Soc. Lond. Ser.~A Math. Phys. Eng. Sci.}}
 \textbf{470} (2014), 20140481, 23~pages, \href{https://arxiv.org/abs/1406.4585}{arXiv:1406.4585}.

\bibitem{KQ}
van~der Kamp P.H., Quispel G.R.W., The staircase method: integrals for periodic
 reductions of integrable lattice equations, \href{https://doi.org/10.1088/1751-8113/43/46/465207}{\textit{J.~Phys.~A: Math. Theor.}}
 \textbf{43} (2010), 465207, 34~pages, \href{https://arxiv.org/abs/1005.2071}{arXiv:1005.2071}.

\bibitem{ves_shab}
Veselov A.P., Shabat A.B., A dressing chain and the spectral theory of the
 {S}chr\"odinger operator, \href{https://doi.org/10.1007/BF01085979}{\textit{Funct. Anal. Appl.}} \textbf{27} (1993),
 81--96.

\bibitem{WE}
Wahlquist H.D., Estabrook F.B., B\"acklund transformation for solutions of the
 {K}orteweg--de {V}ries equation, \href{https://doi.org/10.1103/PhysRevLett.31.1386}{\textit{Phys. Rev. Lett.}} \textbf{31} (1973),
 1386--1390.

\bibitem{ZPZ}
Zhang C., Peng L., Zhang D.-J., Discrete {C}rum's theorems and integrable
 lattice equations, \href{https://arxiv.org/abs/1802.10044}{arXiv:1802.10044}.

\end{thebibliography}
\end{document}